\documentclass[conference]{IEEEtran}
\IEEEoverridecommandlockouts
\usepackage{graphicx}
\usepackage{array}
\usepackage{amsmath}
\usepackage{float}
\usepackage{multirow}
\usepackage{array}
\usepackage{cancel}
\usepackage{tabu}
\usepackage{amsthm}
\usepackage{epsfig}
\usepackage{epstopdf}
\usepackage{epsf}
\usepackage{color}
\usepackage[english]{babel}

\theoremstyle{definition}
\newtheorem{theorem}{Theorem}
\newtheorem{lemma}[theorem]{Lemma}
\newtheorem{definition}{Definition}
\newtheorem{proposition}{Proposition}

\usepackage{amssymb}
\usepackage{cite}
\usepackage{amsmath,amssymb,amsfonts}
\usepackage{graphicx}
\usepackage{textcomp}
\usepackage{xcolor}
\usepackage{algorithm}
\usepackage[noend]{algpseudocode}
\usepackage[T1]{fontenc}
\usepackage{graphicx,lipsum,afterpage,subcaption}
\usepackage{enumitem}
\usepackage{lipsum}
\usepackage{mathtools}
\usepackage{cuted}
\def\BibTeX{{\rm B\kern-.05em{\sc i\kern-.025em b}\kern-.08em
    T\kern-.1667em\lower.7ex\hbox{E}\kern-.125emX}}
\begin{document}

\title{Minimizing Impurity Partition Under Constraints\\
}

	\author{\IEEEauthorblockN{Thuan Nguyen}
		\IEEEauthorblockA{School of Electrical and\\Computer Engineering\\
			Oregon State University\\
			Corvallis, OR, 97331\\
			Email: nguyeth9@oregonstate.edu}

		\and
		\IEEEauthorblockN{Thinh Nguyen}
		\IEEEauthorblockA{School of Electrical and\\Computer Engineering\\
			Oregon State University\\
			Corvallis, 97331 \\
			Email: thinhq@eecs.oregonstate.edu}
	}

\maketitle

\begin{abstract}
Set partitioning is a key component of many algorithms in machine learning, signal processing and communications. In general, the problem of finding a partition that minimizes a given impurity (loss function) is NP-hard.  As such, there exists a wealth of literature on approximate algorithms and theoretical analyses of the partitioning problem under different settings.  In this paper, we formulate and solve a variant of the partition problem called the minimum impurity partition under constraint (MIPUC).  MIPUC finds an optimal partition that minimizes a given loss function under a given concave constraint. MIPUC generalizes the recently proposed deterministic information bottleneck problem which finds an optimal partition that maximizes the mutual information between the input and partitioned output while minimizing the partitioned output entropy.  Our proposed algorithm is developed based on a novel optimality condition, which allows us to find a locally optimal solution efficiently. Moreover, we show that the optimal partition produces a hard partition that is equivalent to the cuts by hyper-planes in the probability space of the posterior probability that finally yields a polynomial time complexity algorithm to find the globally optimal partition. Both theoretical and numerical results are provided to validate the proposed algorithm. 
\end{abstract}

\vspace{-0.05 in}
\section{Introduction}
Partitioning algorithms play a key role in machine learning, signal processing and communications.  Given a set $\mathbb{M}$ consisting of $M$ $N$-dimensional elements and a loss function over the subsets of  $\mathbb{M}$, a $K$-optimal partition algorithm splits $\mathbb{M}$ into $K$ subsets such that the total loss over all $K$ subsets is minimized. The loss function has also termed the impurity which measures the "impurity" of the set. Some of the popular impurity functions are the entropy function and the Gini index \cite{quinlan2014c4}. For example, when the empirical entropy of a set is large, this indicates a high level of non-homogeneity of the elements in the set, i.e., "impurity".  Thus, a $K$-optimal partition algorithm divides the original set into $K$ subsets such that the weighted sum of entropies in each subset is minimal.  

In general, the partitioning problem is NP-hard.  For small $M$, $N$, and $K$, the optimal partition can be found using an exhaustive search with time complexity $O(K^M)$.  In some special cases such as when $N = 2$, and a particular form of impurity functions is used, it is possible to determine the optimal partition in $O(M\log{M})$, independent of $K$ \cite{breiman2017classification}.  On the other hand, for large $M$, $N$,  and $K$, exhaustive  search is infeasible, and it is necessary to use approximate algorithms.  To that end, several heuristic algorithms are commonly used \cite{nadas1991iterative}, \cite{chou1991optimal}, \cite{coppersmith1999partitioning}, \cite{burshtein1992minimum} to find the optimal partition. These algorithms exploit the property of the impurity function to reduce the time complexity.  Specially, in   \cite{coppersmith1999partitioning}, \cite{burshtein1992minimum}, a class of impurity function called "frequency weighted concave impurity" is investigated. Both Gini index and entropy function belong to the frequency weighted concave impurity class.  Furthermore, assuming the concavity of the impurity function, Brushtein et al.  \cite{burshtein1992minimum} and Coppersmith et al. \cite{coppersmith1999partitioning} showed that the optimal partition can be separated by a hyper-plane in the probability space. Consequently, they proposed approximate algorithms to find the optimal partition. Recently, in \cite{laber2018binary}, an approximate algorithm is proposed for a binary partition ($K=2$) that guarantees the true impurity is within a constant factor of the approximation.
From a communication/coding theory perspective, the problem of finding an optimal quantizer that maximizes the mutual information between the input and the quantized output is an important instance of the partition problem.  In particular, algorithms for constructing polar codes \cite{tal2013construct} and for decoding LDPC codes \cite{romero2015decoding} made use of the quantizers. Consequently, there has been recent works on designing quantizers for maximizing mutual information \cite{kurkoski2014quantization}, \cite{nguyen2018capacities}.  

In this paper, we extended the problem of minimizing impurity partition under the constraints of the output variable. It is worth noting that many of  problem in the real scenario is the optimization under constraints, therefore, our extension problem is interesting and applicable. For example, our setting generalizes the recently proposed deterministic information bottleneck \cite{strouse2017deterministic} that finds the optimal partition to maximize the mutual information between input and quantized output while keeps the output entropy is as small as possible. It is worth noting that Strouse et al. used a technique which is similar to  the information bottleneck method \cite{tishby2000information} and is hard to extend to other impurity and constraint functions. On the other hand, our proposed method  is developed based on a novel optimality condition, which allows us to find a locally optimal solution efficiently for an arbitrary frequency weighted concave impurity functions under arbitrary concave constraints. Moreover, we show that the optimal partition produces a hard partition that is equivalent to the cuts by hyper-planes in the probability space of the posterior probability that finally yields a polynomial time complexity algorithm to find the globally optimal partition.


\section{Problem Formulation}

	\begin{figure}
		\centering
		\includegraphics[width=1.6 in]{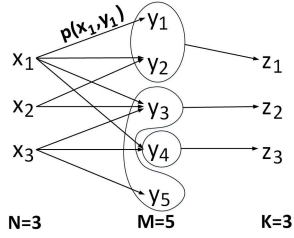}\\
		\caption{$Q(Y) \rightarrow Z$ for a given joint distribution $p_{(X,Y)}$.}\label{fig: 1}
	\end{figure}	

Consider an original discrete data $X_i \in X=\{X_1,X_2,\dots,X_N\}$ with distribution $p_X=[p_1,p_2,\dots,p_N]$ is given. Due to the affection of noise, one only can view a noisy version of data $X$ named $Y=\{Y_1,\dots,Y_M\}$ with the joint probability distribution $p_{(X_i,Y_j)}$ is given $\forall$ $i =1,2,\dots,N$ and $j=1,2,\dots,M$. It is easily to compute the distribution of $Y$, i.e., $p_Y=[q_1,q_2,\dots,q_M]$. Therefore, each sample $Y_i$ is specified by a joint probability distribution vector $p_{(X,Y_i)}=[p_{(X_1,Y_i)},p_{(X_2,Y_i)}, \dots, p_{(X_N,Y_i)}]$ which involves two parameters (i) the probability weight $q_i$ and (ii) a conditional probability vector tuple  $p_{X|Y_i}=[p_{X_1|Y_i},p_{X_2|Y_i},\dots,p_{X_N|Y_i}]$. 
From the discrete data $Y$, the partitioned output $ Z=\{Z_1,\dots,Z_K\}$ with the distribution $p_Z=[v_1,v_2,\dots,v_K ]$ is obtained by applying an quantizer (possible stochastic) $Q$ which assigns $Y_j \in Y$ to a partitioned output subset $Z_i \in Z$ by a probability $p_{Z_i|Y_j}$ where $0 \leq p_{Z_i|Y_j} \leq 1$.
\begin{equation}
Q( Y) \rightarrow Z.
\end{equation}

Fig. \ref{fig: 1} illustrates our setting.  Our goal is finding an optimal quantizer (partition) $Q^*$ such that the impurity  function  $F(X,Z)$ between original data $X$ and partitioned output $Z$ is minimized while the partitioned output probability distribution $p_Z=[v_1,v_2,\dots,v_K]$ satisfies a constraint $C(p_Z) \leq D$. 

\subsection{Impurity measurement}  
The impurity $F(X,Z)$ between $X$ and $Z$ is defined by adding up the impurity in each output subset $Z_i \in Z$ i.e., $F(X,Z) = \sum_{i=1}^{K} F(X,Z_i)$ 
where
\begin{eqnarray}
\label{eq: definition of impurity}
F(X,Z_i) &\!=\!&  p_{Z_i} f[p_{X|Z_i}] \nonumber\\
&\!=\!&  v_i f[p_{X_1|Z_i},\! p_{X_2|Z_i},\! \dots,\! p_{X_N|Z_i}]
\label{eq: definition of impurity}
\end{eqnarray}
is the impurity function in $Z_i$, $p_{X|Z_i}=[p_{X_1|Z_i},\! p_{X_2|Z_i},\! \dots,\! p_{X_N|Z_i}]$ denotes the conditional distribution $p_{X|Z_i}$. The loss function $f(.)$ is a concave function which is defined as following.

\begin{definition}
\label{def: 1}
A concave loss function $f(.)$ is a real function in $\mathbf{R^N}$ such that: 

(i) For all probability vector $a=[a_1,a_2,\dots,a_N]$ and $b=[b_1,b_2,\dots,b_N]$ 
\begin{equation}
\label{eq: concave function}
f(\lambda a + (1-\lambda)b) \geq \lambda f(a) + (1-\lambda)f(b), \forall \lambda \in (0,1)
\end{equation}
with equality if and only if $a=b$.

(ii) $f(a)=0$ if $a_i=1$ for some $i$. 
\end{definition}

We note that the above definition of impurity function was proposed in \cite{chou1991optimal}, \cite{coppersmith1999partitioning}, \cite{burshtein1992minimum}. Many of interesting impurity functions such as Entropy and Gini index  \cite{chou1991optimal}, \cite{coppersmith1999partitioning}, \cite{burshtein1992minimum} satisfy the Definition \ref{def: 1}. \\
%

\textbf{Reformulation of the impurity function:} We will show that the impurity function $F(X,Z_i)$ can be rewritten as the function of only the joint distribution variable $p_{(X,Z_i)}=[p_{(X_1,Z_i)}, p_{(X_2,Z_i)}, \dots, p_{(X_N,Z_i)}]$. Therefore, one can denote $F(X,Z_i)$ as $F(p_{(X,Z_i)})$. Indeed, define $$p_{(X_j,Z_i)}=\sum_{Y_k \in Y}^{}p_{(X_j,Y_k)} p_{Z_i|Y_k}. $$
Now,  the impurity function can be rewritten by:
\begin{eqnarray}
 F(X,Z_i) &\!=\!&  \sum_{j\!=\!1}^{N}p_{(X_j\!,\!Z_i)}f[\dfrac{p_{(X_1\!,\!Z_i)}}{\sum_{j\!=\!1}^{N}p_{(X_j\!,\!Z_i)}}, \dots, \dfrac{p_{(X_N\!,\!Z_i)}}{\sum_{j\!=\!1}^{N}p_{(X_j\!,\!Z_i)}}] \nonumber\\
\label{eq: new formulation}
\end{eqnarray}
where $\sum_{j=1}^{N}p_{(X_j,Z_i)}=v_i$ denotes the weight of $Z_i$ and $\dfrac{p_{(X_k,Z_i)}}{\sum_{j=1}^{N}p_{(X_j,Z_i)}}$ denotes the conditional distribution $p_{(X_k|Z_i)}$.  
\textit{The impurity function $F(X,Z_i)$, therefore, is a function of $p_{(X_j,Z_i)}$ variables. In the rest of this paper, we will denote $F(X,Z_i)$ by $F(p_{(X,Z_i)})$ and $F(X,Z)$ by $F(p_{(X,Z)})$.} 


\subsection{Partitioned output constraint}
Now, we formulate a new problem such that the impurity function is minimized  while the partitioned output distribution  $p_Z=[v_1,v_2,\dots,v_K]$ satisfies a constraint. 
\begin{equation*}
C(p_Z)=g(v_1) + g(v_2) + \dots + g(v_K) \leq D
\end{equation*}
where $g(.)$ is a concave function. For example, 

\begin{itemize}
\item{} Entropy function:
\begin{equation*}
H[p_Z]=H[v_1,v_2,\dots,v_K]=-\sum_{i=1}^{n}v_i \log(v_i).
\end{equation*}
For example, if we want to compress data $Y$ to $Z$ and then transmit $Z$ as the intermediate representation of $Y$ over a low bandwidth channel to the next destination, the entropy of $p_Z$ which is controlled the maximum compression rate, is important. A lower of  $H(p_Z)$, a smaller of channel capacity is required \cite{strouse2017deterministic}.  
\item{} Linear function:
Similar to previous example, to transmit $Z$ over a channel, each value in the same subset $Z_1,Z_2,\dots,Z_K$ is coded to a pulse, i.e., $Z_1 \rightarrow 0$, $Z_2 \rightarrow 1$, $Z_3 \rightarrow 2$ which have a difference cost of transmission i.e., power consumption or time delay. The cost of transmission now is
\begin{equation*}
T[p_Z]=T[v_1,v_2,\dots,v_K]=\sum_{i=1}^{K}t_iv_i.
\end{equation*}
where $t_i$ is a constant relate to power consumption or time delay. An example of transmission cost can be viewed in  \cite{verdu1990channel}.
\end{itemize}
 
\subsection{Problem Formulation}

Now, our problem can be formulated as finding an optimal quantizer  $Q^*$ such that the impurity function  $F(X,Z)$ is minimized while the partitioned output probability distribution $p_Z$ satisfies a constraint $C(p_Z) \leq D$. 
%
Since both $F(X,Z)$ and $C(p_Z) $ depend on the quantizer design, we are interested in solving the following optimization problem
	\begin{equation}
	\label{eq: main problem}
	Q^*=\min_{Q}[\beta F(X,Z) + C(p_Z) ],
	\end{equation}
	where $\beta$ is pre-specified parameter to control a given trade-off between minimizing $F(X,Z)$ or $C(p_Z)$. 

\textbf{Relate to Deterministic Information Bottleneck (DIB) method:} we also note that our optimization problem in (\ref{eq: main problem}) covers the proposed problem called Deterministic Information Bottleneck Method \cite{strouse2017deterministic} which solved the following problem
	\begin{equation}
	\label{eq:determinisitc information bottleneck problem}
	Q^*=\min_{Q}[H(Z)- \beta I(X;Z) ],
	\end{equation}
where $H(Z)$ is the entropy of output $Z$ and $I(X;Z)$ is the mutual information between original data $X$ and  quantized output $Z$. Minimizing $H(Z)$ is equivalent to minimizing $C(p_Z)$. Moreover,
\begin{equation*}
I(X;Z)=H(X)-H(X|Z).
\end{equation*}
Thus, minimizing $-I(X;Z)$ is equivalent to minimizing $H(X|Z)$ due to $p_X$ is given.  That said Deterministic Information Bottleneck \cite{strouse2017deterministic}  is a special case of our problem where both $f(.)$ and $g(.)$ are entropy functions.

\section{Solution approach}
\subsection{Optimality condition}

We first begin with some properties of the impurity function. For convenience, we recall that $F(p_{(X,Z_i)})$ denotes the impurity function in output subset $Z_i$ and $p_{X|Z_i}=[\dfrac{p_{(X_1\!,\!Z_i)}}{\sum_{j=1}^{N}p_{(X_j,Z_i)}}, \dots, \dfrac{p_{(X_N,Z_i)}}{\sum_{j=1}^{N}p_{(X_j,Z_i)}}]$.
\begin{proposition}
\label{prop: 2}
The impurity function $F(p_{(X,Z_i)})$ in partitioned output $Z_i$ has the following properties:

(i) \textbf{proportional increasing/ decreasing to its weight:} if  $p_{(X,Z_i)}=\lambda p_{(X,Z_j)}$, then
\begin{equation}
\dfrac{F(p_{(X,Z_i)})}{F(p_{(X,Z_j)})}=\lambda.
\end{equation}

(ii)\textbf{ impurity gain after partition is always non-negative:} If $p_{(X,Z_i)}=p_{(X,Z_j)}+p_{(X,Z_k)}$, then

\begin{equation}
\label{eq: concave of partition}
F(p_{(X,Z_i)}) \geq F(p_{(X,Z_j)}) + F(p_{(X,Z_k)}).
\end{equation}

\end{proposition}
\begin{proof}

(i) From $p_{(X,Z_i)}=\lambda p_{(X,Z_j)}$, then  $p_{X|Z_i}=p_{X|Z_j}$ and $p_{Z_i}=\lambda p_{Z_j}$. Thus, using the definition of $F(p_{(X,Z_i)})$ in (\ref{eq: definition of impurity}), it is obviously to prove the first property. 

(ii) By dividing both side of  $p_{(X,Z_i)}=p_{(X,Z_j)}+p_{(X,Z_k)}$ to $p_{Z_i}$, we have
\begin{equation}
\label{eq: 9}
p_{X|Z_i}= \dfrac{p_{Z_j}}{p_{Z_i}}p_{X|Z_j}+ \dfrac{p_{Z_k}}{p_{Z_i}}p_{X|Z_k}.
\end{equation}
Now, using the definition of $F(X,Z_i)$ in (\ref{eq: definition of impurity}),
\begin{eqnarray}
F(\!p_{(X,Z_i)}\!) &\!=\!& p_{Z_i}f(p_{X|Z_i}) \nonumber\\
&\!=\!& p_{Z_i} f [\dfrac{p_{Z_j}}{p_{Z_i}}p_{X|Z_j} \!+\! \dfrac{p_{Z_k}}{p_{Z_i}}p_{X|Z_k}] \label{eq: 10}\\
&\!\geq \!& p_{\!Z_i\!} [\dfrac{p_{Z_j}}{p_{Z_i}} f(\!p_{X|Z_j}\!) \!+\! \dfrac{p_{Z_k}}{p_{Z_i}}f(\!p_{X|Z_k}\!)] \label{eq: 11}\\
&\!=\!&p_{Z_j}f(p_{X|Z_j}) \!+\! p_{Z_k}f(p_{X|Z_k}) \nonumber\\
&\!=\!& F(p_{(X,Z_j)}) + F(p_{(X,Z_k)}) \nonumber
\end{eqnarray}
with (\ref{eq: 10}) is due to (\ref{eq: 9}) and (\ref{eq: 11}) due to concave property of $f(.)$ which is defined in (\ref{eq: concave function}) using $\lambda= \dfrac{p_{Z_j}}{p_{Z_i}}$, $1-\lambda=\dfrac{p_{Z_k}}{p_{Z_i}}$. 
\end{proof}

Now, we are ready to prove the main result which characterizes the condition for an optimal partition $Q^*$. 
\begin{theorem}
\label{theorem: 1}
Suppose that  an optimal partition $Q^*$ yields the optimal partitioned output $Z=\{Z_1,Z_2,\dots,Z_K \}$.
 For each optimal subset $Z_l$, $l \in \{1,2,\dots,K\}$, we define vector $c_l=[c_l^1,c_l^2,\dots,c_l^N]$ where
\begin{equation}
\label{eq: 16}
c_l^i= \frac{\partial F(p_{(X,Z_l)})}{\partial p_{(X_i,Z_l)}}, \forall i \in \{1,2,\dots,N\}. 
\end{equation}
We also define 
\begin{equation}
\label{eq: 17}
d_l=\frac{\partial g(v_l)}{\partial v_l}. 
\end{equation}
Define the "distance" from $Y_i \in Y$ to $Z_l$ is 
\begin{eqnarray}
D(Y_i,Z_l) &=&\beta \sum_{q=1}^{N}[p_{(X_q,Y_i)} c_l^q] + d_l q_i \nonumber\\
			   &=& q_i (\beta \sum_{q=1}^{N}[p_{X_q|Y_i} c_l^q] + d_l) \label{eq: optimality condition}.
\end{eqnarray}
Then, data $Y_i$ with probability $q_i$  is quantized to $Z_l$ if and only if $D(Y_i,Z_l) \leq D(Y_i,Z_s)$ for $\forall s \in \{1,2,\dots,K\}$, $ s \neq l$. 
\end{theorem}
\begin{proof}

	\begin{figure}
		\centering
		\includegraphics[width=2.7 in]{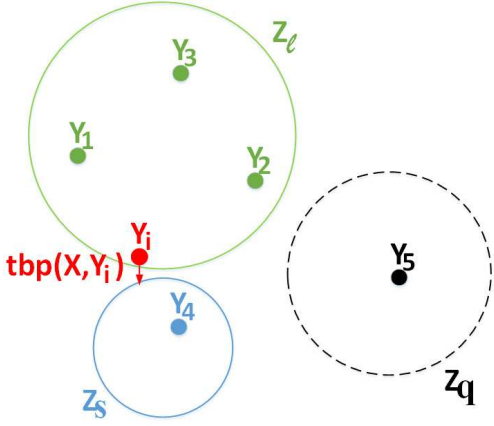}\\
		\caption{"Soft" partition of $Y_i$ between $Z_l$ and $Z_s$ by changing amount of $tbp_{(X,Y_i)}$.}\label{fig: 2}
	\end{figure}	

Now, consider two arbitrary partitioned outputs $Z_l$ and $Z_s$ and a trial data $Y_i$. For a given optimal quantizer $Q^*$, we suppose that $Y_i$ is allocated to $Z_l$ with the probability of $p_{Z_l|Y_i}=b$, $0 < b \leq 1$. We remind that $p_{(X,Y_i)}=[p_{(X_1,Y_i)},p_{(X_2,Y_i)}, \dots, p_{(X_N,Y_i)}]$ denotes the joint distribution of sample $Y_i$. We consider the change of the impurity function $F(p_{(X,Z)})$ and the constraint  $C(p_Z)$ as a function of $t$ by changing amount $tbp_{(X,Y_i)}$ from $p_{(X,Z_l)}$ to $p_{(X,Z_s)}$ where $t$ is a scalar and $0 \leq t \leq 1$. 
\begin{small}
\begin{eqnarray}
F(p_{(X\!,\!Z)})(t) &\!=\!&  \sum_{q=1, q \neq l,s}^{K}F(p_{(X,Z_q)}) \nonumber \\
&\!+\!&  F(p_{(X,Z_s)} \!+\! tbp_{(X,Y_i)}) \!+\! F(p_{(X,Z_l)} \!-\! tbp_{(X,Y_i)}),\nonumber \\
\label{eq: 18}
\end{eqnarray}
\begin{eqnarray}
C(p_Z)(t) &=& \sum_{q=1,q \neq l,s}^{K}g(p_{Z_q}) + g(p_{Z_l}-tbq_i) + g(p_{Z_s} + tbq_i), \nonumber \\
\label{eq: 19}
\end{eqnarray}
\end{small}
where $p_{(X,Z_s)} + tbp_{(X,Y_i)}$ and $p_{(X,Z_l)} - tbp_{(X,Y_i)}$ denotes the new joint distributions in $Z_s$  and $Z_l$ by changing amount of $tbp_{(X,Y_i)}$ from $Z_l$ to $Z_s$. 
Fig. \ref{fig: 2} illustrates our setting. From (\ref{eq: 18}) and (\ref{eq: 19}), the total instantaneous change of $\beta F(p_{(X,Z)}) + C(p_Z)$  by changing amount of $tbp_{(X,Y_i)}$ is
\begin{eqnarray}
\label{eq: I(t)}
I(t) &=& \beta [F(p_{(X,Z_s)} \!+\! tbp_{(X,Y_i)}) \!+\! F(p_{(X,Z_l)} \!-\! tbp_{(X,Y_i)})] \nonumber\\
&+&  g(p_{Z_s} + tbq_i) + g(p_{Z_l}-tbq_i).
\end{eqnarray}
However,
\begin{small}
 \begin{equation}
 \label{eq: derivative 1}
 \frac{\partial F(p_{(X,Z_l)} \!-\! tbp_{(X,Y_i)})}{\partial t}|_{t\!=\!0} \!=\! \!-\! b\sum_{q=1}^{N}[c_l^q p_{(X_q,Y_i)}] \!=\!-\! q_i b \sum_{q=1}^{N}[c_l^q p_{X_q|Y_i}].
 \end{equation}
 \begin{equation}
  \label{eq: derivative 2}
   \frac{\partial F(p_{(X,Z_s)} \!+\! tbp_{(X,Y_i)})}{\partial t}|_{t\!=\!0}=b\sum_{q=1}^{N}[c_s^q p_{(X_q,Y_i)}]=q_i b \sum_{q=1}^{N}[c_s^q p_{X_q|Y_i}].
 \end{equation}
\end{small}
Similarly,
 \begin{equation}
  \label{eq: derivative 3}
   \frac{\partial g(p_{Z_s} + tbq_i) + g(p_{Z_l}-tbq_i) }{\partial t}|_{t=0}=b(d_sq_i-d_lq_i).
 \end{equation}
From (\ref{eq: 16}), (\ref{eq: 17}), (\ref{eq: derivative 1}), (\ref{eq: derivative 2})  and (\ref{eq: derivative 3}), we have
\begin{eqnarray}
\frac{\partial I(t)}{\partial t}|_{t=0} &=& b \beta \sum_{q=1}^{N}[c_s^q p_{(X_q,Y_i)}] + b d_sq_i \nonumber  \\ 
&-& b \beta \sum_{q=1}^{N}[c_l^q p_{(X_q,Y_i)}] - b d_lq_i \nonumber  \\
&=& b q_i [\beta \sum_{q=1}^{N}c_s^q p_{X_q|Y_i}+d_s]  \nonumber \\
&-& b q_i [\beta \sum_{q=1}^{N}c_l^q p_{X_q|Y_i} +d_l] \nonumber\\
&=&b (D(Y_i,Z_s)-D(Y_i,Z_l)). \nonumber
\end{eqnarray}
Now, using contradiction method, suppose that $D(Z_l,Y_i) > D(Z_s,Y_i)$. Thus,
\begin{equation}
\label{eq: 20}
\frac{\partial I(t)}{\partial t} |_{t=0} < 0.
\end{equation}
\begin{proposition}
\label{prop: 1}
Consider $I(t)$ which is defined in (\ref{eq: I(t)}). For $0 < t < a < 1$, we have:
\begin{equation}
\label{eq: gradient of I(t)}
\dfrac{I(t)-I(0)}{t} \geq \dfrac{I(a)-I(0)}{a}.
\end{equation}
\end{proposition}
\begin{proof}
From Proposition \ref{prop: 2},
\begin{footnotesize}
\begin{eqnarray}
F(p_{(X,Z_s)} \!+\! tbp_{(X,Y_i)}) &\!\geq\!& F((1\!-\!\dfrac{t}{a}) p_{(X,Z_s)}) \!+\! F(\dfrac{t}{a}(p_{(X,Z_s)} \!+\! abp_{(X,Y_i)})) \nonumber\\
&\!=\!& (1\!-\!\dfrac{t}{a}) F(p_{(X,Z_s)}) \!+\! \dfrac{t}{a}F(p_{(X,Z_s)} \!+\! abp_{(X,Y_i)}).\nonumber \\\label{eq: 21}
\end{eqnarray}
\begin{eqnarray}
F(p_{(X,Z_l)} \!-\! tbp_{(X,Y_i)}) &\!\geq\!& F((1-\dfrac{t}{a})p_{(X,Z_l)})+F(\dfrac{t}{a}(p_{(X,Z_l)} \!-\! abp_{(X,Y_i)})) \nonumber\\
&\!=\!& (1\!-\!\dfrac{t}{a}) F(p_{(X,Z_l)}) \!+\! \dfrac{t}{a}F(p_{(X,Z_l)} \!-\! abp_{(X,Y_i)}),\nonumber \\ \label{eq: 22}
\end{eqnarray}
\end{footnotesize}
where the inequality due to (ii) and the equality due to (i) in Proposition \ref{prop: 1}, respectively. Similar, since $g(.)$ is a  concave function,
\begin{eqnarray}
g(p_{Z_s}\!+\!tbq_i) &\!=\!& g((1-\dfrac{t}{a})p_{Z_s}  \!+\! \dfrac{t}{a} (p_{Z_s}+abq_i) ) \nonumber\\
&\!\geq\! & (1-\dfrac{t}{a})g(p_{Z_s}) \!+\! \dfrac{t}{a} g(p_{Z_s}+abq_i), \nonumber \\ \label{eq: 24}
\end{eqnarray}
\begin{eqnarray}
g(p_{Z_l}\!-\!tbq_i) &\!=\!& g((1-\dfrac{t}{a})p_{Z_l}  \!+\! \dfrac{t}{a} (p_{Z_l}-abq_i) ) \nonumber\\
&\!\geq\!& (1-\dfrac{t}{a})g(p_{Z_l}) \!+\! \dfrac{t}{a} g(p_{Z_l}-abq_i). \nonumber \\ \label{eq: 25}
\end{eqnarray}
Thus, adding up (\ref{eq: 21}), (\ref{eq: 22}), (\ref{eq: 24}), (\ref{eq: 25}) and using a little bit of algebra, one can show that 
\begin{equation}
I(t) \geq (1-\dfrac{t}{a})I(0) + \dfrac{t}{a}I(a)
\end{equation}
which is equivalent to (\ref{eq: gradient of I(t)}). 
\end{proof}
Now, we continue to the proof of Theorem \ref{theorem: 1}. From Proposition \ref{prop: 1} and the assumption in (\ref{eq: 20}), we have:
\begin{equation*}
 0 > \frac{\partial I(t)}{\partial t} |_{t=0}=\lim \dfrac{I(t)-I(0)}{t} \geq \dfrac{I(1)-I(0)}{1}.
\end{equation*}

Thus, $I(0)>I(1)$. That said, by completely changing all $bp_{(X,Y_i)}$ from $Z_l$ to $Z_s$, the total of the impurity is obviously reduced. This contradicts to our assumption that the quantizer $Q^*$ is optimal. By contradiction method, the proof is complete.
\end{proof}

\begin{lemma}
\label{lemma: 1}
The optimal solution to the problem (\ref{eq: main problem}) is a
deterministic quantizer (hard clustering) i.e., $p_{Z_i|Y_j} \in \{0,1\}$, $\forall$ $i,j$.  
\end{lemma}

\begin{proof}
Lemma \ref{lemma: 1} directly follows by the proof of Theorem \ref{theorem: 1}. Since the distance function $D(Y_i,Z_l)$ does not depend on the soft partition $p_{Z_l|Y_i}=b$, one should completely allocate $Y_i$ to a $Z_l$ such that $D(Y_i,Z_l)$ is minimized. 
\end{proof}

\subsection{Practical Algorithm}
Theorem \ref{theorem: 1} gave an optimality condition such that  the "distance" from a data $Y_i$ to its optimal partition $Z_l$ should be shortest. Therefore, a simple algorithm which is similar to a k-means algorithm can be applied to find the locally optimal solution. Our algorithm is proposed in Algorithm \ref{alg: 1}. We also note that the distance from $Y_i$ to $Z_l$ is defined by
\begin{eqnarray}
D(Y_i,Z_l) &=&\beta \sum_{q=1}^{N}[p_{(X_q,Y_i)} c_l^q] + d_l q_i \nonumber\\
			   &=& q_i [\beta \sum_{q=1}^{N}[p_{X_q|Y_i} c_l^q] + d_l].
\end{eqnarray}
Therefore, one can ignore the constant $q_i$ while comparing the distances between $D(Y_i,Z_l)$ and $D(Y_i,Z_s)$. 
 
\begin{footnotesize}
	\begin{algorithm}
		\caption{Finding the optimal partition under partitioned output constraint}
		\label{alg: 1}
		\begin{algorithmic}[1]
			\State{\textbf{Input}: $p_X$, $p_Y$, $p_{(X,Y)}$, $f(.)$, $g(.)$ and $\beta$.}
			\State{\textbf{Output}: $Z=\{Z_1,Z_2,\dots,Z_K$ \} }
			
			\State{\textbf{Initialization}: Randomly hard clustering $Y$ into $K$ clusters. }
			
			\State{\textbf{Step 1}: Updating $p_{(X,Z_l)}$ and $d_l$ for output subset $Z_l$ for $\forall$ $l \in \{1,2,\dots,K\}$:}			
		\begin{equation*}
		p_{(X,Z_l)}= \sum_{Y_q \in Z_l}^{}p_{(X,Y_q)},
		\end{equation*}
		\begin{equation*}
		c_l^i= \frac{\partial F(p_{(X,Z_l)})}{\partial p_{(X_i,Z_l)}}, \forall i \in \{1,2,\dots,N\}, 
		\end{equation*}		
		\begin{equation*}
		v_l =\sum_{Y_q \in Z_l}p_{Y_q},
		\end{equation*}			
		\begin{equation*}
		d_l= \frac{\partial g(v_l)}{\partial v_l}. 
		\end{equation*}
			\State{\textbf{Step 2}: Updating the membership by measurement the distance from each $Y_i \in Y$ to each subset $Z_j \in Z$}
\begin{equation}
\label{eq: nearest clustering}
Z_l=\{Y_i| D(Y_i,Z_l) \leq D(Y_i,Z_s), \forall s \neq l.
\end{equation}
			\State{\textbf{Step 3}: Go to Step 1 until the partitioned output $\{Z_1,Z_2,\dots,Z_K\}$ stop changing or the maximum number of iterations has been reached.}				
		\end{algorithmic}
	\end{algorithm}
\end{footnotesize}

The Algorithm \ref{alg: 1} works similarly to k-means algorithm and the distance from each point in $Y$ to each partition subset in $Z$ is updated in each loop. The complexity of this algorithm, therefore, is $O(TNKM)$ where $T$ is the number of iterations, $N$, $K$, $M$ are the size of the data dimensional, the output size and the data size.

\subsection{Hyperplane separation}
\label{subsection: 3-C}

	\begin{figure}
		\centering
		\includegraphics[width=2 in]{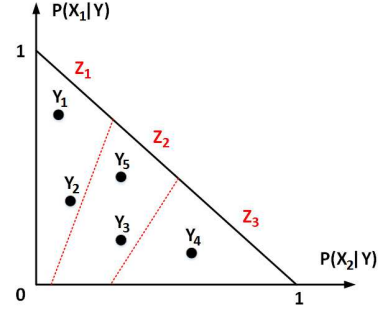}\\
		\caption{For $N=3$, $M=5$ and $K=3$, the optimal quantizer is equitvalent to the hyperplane cuts in 2 dimensional probability space.}\label{fig: 4}
	\end{figure}	

Similar to the work in \cite{burshtein1992minimum}, we show that the optimal partitions correspond to the regions separated by
hyper-plane cuts in the probability space of the posterior distribution. Consider the optimal quantizer $Q^*$ that produces a given partitioned output sets $Z=\{Z_1,Z_2,\dots,Z_K\}$ and a given conditional probability $p_{X|Z_l}=\{ p_{X_1|Z_l}, \dots, p_{X_N|Z_l}\}$ for $\forall$ $l=1,2,\dots,K$. From the optimality condition in Theorem \ref{theorem: 1}, we know that $\forall$ $Y_i \in Z_l$, then 
\begin{equation*}
\beta \sum_{q=1}^{N}[p_{X_q|Y_i} c_l^q] + d_l \leq \beta \sum_{q=1}^{N}[p_{X_q|Y_i} c_s^q] + d_s.
\end{equation*}
Thus, $ 0 \geq   \beta \sum_{q=1}^{N}p_{X_q|Y_i}[c_l^q-c_s^q] + d_l-d_s.$
By using $p_{X_N|Y_i}=1-\sum_{q=1}^{N-1} p_{X_q|Y_i}$, we have
\begin{small}
\begin{eqnarray}
d_s\!-\!d_l\!+\!\beta(c_s^N\!-\!c_l^N) & \! \geq \! &  \sum_{q=1}^{N-1} \beta p_{X_q|Y_i}[c_l^q \!-\! c_s^q \!-\! c_l^N+c_s^N].\nonumber \\ \label{eq: hyperplane}
\end{eqnarray}
\end{small}
For a given optimal quantizer $Q^*$, $c_l^q$ ,$c_s^q$, $d_l$, $d_s$ are scalars and $0 \leq p_{X_q|Y_i} \leq 1$, $\sum_{q=1}^{N}p_{X_q|Y_i}=1$. From (\ref{eq: hyperplane}),  $Y_i \in Z_l$ belongs to a region separated by a hyper-plane cut in probability space of posterior distribution $p_{X|Y_i}$. Similar to the result proposed in \cite{burshtein1992minimum}, existing a polynomial time  algorithm having time complexity of $O(M^{N})$ that can determine the globally optimal solution for the problem in (\ref{eq: main problem}). Fig. \ref{fig: 4} illustrates the hyper-plane cuts in two dimensional probability space for $N=3,M=5$ and $K=3$. 
\subsection{Application}
\label{sec: application}

As discussed in the previous part  that the Deterministic Information Bottleneck \cite{strouse2017deterministic} is a special case of our problem for the impurity function and the output constraints are entropy functions. 
We refer reader to \cite{strouse2017deterministic} for more detail of applications.  In this paper, we want to provide a simple example that using the results in Sec. \ref{subsection: 3-C} to find the globally optimal quantizer for a binary input communication channel quantization. Fig. \ref{fig: 5} illustrates our application. Output $Z$ is quantized from data $Y$. Next, $Z$ is mapped to $W$ by a mapping function $W=f(Z)$. Now, $W$ is the input for a limited rate channel $C$. Our goal is to design a good quantizer such that the mutual information $I(X;Z)$ has remained as much as possible while the rate of output $Z$ is under the limited rate $C$. We also note that a similar constraint, i.e., cost transmission, time delay can be replaced  to formulate other interesting problems. 
	\begin{figure}
		\centering
		\includegraphics[width=3.5 in]{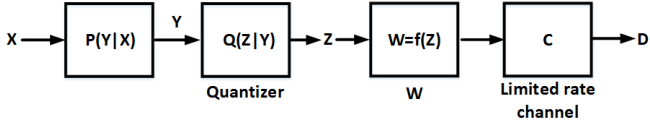}\\
		\caption{Partition with output contraints: an example with a relay channel having a limited capacity.}\label{fig: 5}
	\end{figure}

\textbf{Example 1:} To illustrate how the Algorithm \ref{alg: 1} work, we provide the following example. Consider a communication system which transmits input $X=\{X_1=-1,X_2=1\}$ having $p_{X_1}=0.2$, $p_{X_2}=0.8$ over an additive noise channel with i.i.d noise distribution $N(\mu=0,\sigma=1)$. The output signal  $Y$ is a continuous signal which is the result of input $X$ adding to the noise $N$.
\vspace{-0.05 in}
$$Y = X + N.$$
Due to the additive property, the conditional distribution of output $Y$ given input $X_1$ is $p_{Y|X_1=-1}=N(-1,1)$ while the conditional distribution of output $Y$ given input $X_2$ is $p_{Y|X_2=1}=N(1,1)$. We also note that due to the additive noise is continuous, $Y$ is in continuous domain.  
The continuous output $Y$ then is quantized to binary output $Z=\{Z_1=-1,Z_2=1\}$ using a quantizer $Q$. Quantized output $Z$ is transmitted over a limited rate channel $C$ with the highest rate $R=0.5$.  We have to find an optimal quantizer $Q^*$ such that the mutual information $I(X;Z)$ is maximized while $H(Z) \leq R$. 
Now, we first discrete $Y$ to $M=200$ pieces from $[-10,10]$ with the same width $\epsilon=0,1$. Thus, $Y=\{Y_1,Y_2,\dots,Y_{200}\}$ with the joint distribution $p_{(X_i,Y_j)}$, $i=1,2$, $j=1,2,\dots,200$ can be determined by using two given conditional distributions $p_{Y|X_1=-1}=N(-1,1)$ and $p_{Y|X_2=1}=N(1,1)$. 
Next, to find the optimal quantizer $Q^*$, we scan all the possible value of $\beta \geq 0$. For each value of $\beta$, we run the Algorithm \ref{alg: 1} many times to find the globally optimal quantizer. Finally, the largest mutual information $I(X;Z)^*$ is $0.18623$ which corresponds  to $H(Z)=0.48873$ at $\beta^*=6$.

\textbf{Using hyper-plane separation to find the globally optimal solution:} Using the result in Sec. \ref{subsection: 3-C}, the optimal quantizer (both local and global) is equivalent to a hyper-plane cut in probability space. Due to $|X|=N=2$, the hyper-plane is a scalar in posterior distribution $p_{X_2|Y}$. Noting that $p_{X_1|Y}$ is a strictly increasing function over $Y=[-10,10]$.  Thus, an exhausted searching of $y \in [-10,10]$ can be applied to find the optimal quantizer. Fig. \ref{fig: 6}  illustrates the function of $I(X;Z)$ and $H(Z)$ with variable $y \in [-10;10]$ using the resolution $\epsilon=0.1$. For $\beta=6$, the optimal mutual information $I(X;Z)^*=0.18623$ corresponds to  $H(Z)=0.48873$ that are achieved at $y=-1.1$. This result confirms the globally optimal solution using Algorithm \ref{alg: 1} in Example 1. 

	\begin{figure}
		\centering
		\includegraphics[width=2.8 in]{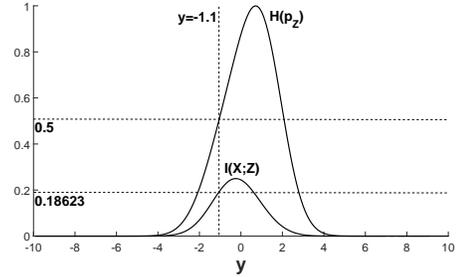}\\
		\caption{Finding the maximum of $I(X;Z)$ under the constraint $H(Z) \leq 0.5$.}\label{fig: 6}
	\end{figure}	

\section{Conclusion}
In this paper, we presented a new framework to minimize the  impurity partition while the probability distribution of partitioned output satisfies a concave constraint. Based on the optimality condition, we show that the optimal partition should be a hard partition. A low complexity algorithm is provided to find the locally optimal solution. Moreover, we show that the optimal partitions (local/global) correspond to the regions separated by hyper-plane cuts in the probability space of the posterior distribution. Therefore, existing a polynomial time complexity algorithm that can find the truly globally optimal solution. 

%

\bibliographystyle{unsrt}
\bibliography{sample}

\begin{thebibliography}{10}

\bibitem{quinlan2014c4}
J~Ross Quinlan.
\newblock {\em C4. 5: programs for machine learning}.
\newblock Elsevier, 2014.

\bibitem{breiman2017classification}
Leo Breiman.
\newblock {\em Classification and regression trees}.
\newblock Routledge, 2017.

\bibitem{nadas1991iterative}
Arthur N{\'a}das, David Nahamoo, Michael~A Picheny, and Jeffrey Powell.
\newblock An iterative'flip-flop'approximation of the most informative split in
  the construction of decision trees.
\newblock In {\em [Proceedings] ICASSP 91: 1991 International Conference on
  Acoustics, Speech, and Signal Processing}, pages 565--568. IEEE, 1991.

\bibitem{chou1991optimal}
Philip~A. Chou.
\newblock Optimal partitioning for classification and regression trees.
\newblock {\em IEEE Transactions on Pattern Analysis \& Machine Intelligence},
  (4):340--354, 1991.

\bibitem{coppersmith1999partitioning}
Don Coppersmith, Se~June Hong, and Jonathan~RM Hosking.
\newblock Partitioning nominal attributes in decision trees.
\newblock {\em Data Mining and Knowledge Discovery}, 3(2):197--217, 1999.

\bibitem{burshtein1992minimum}
David Burshtein, Vincent Della~Pietra, Dimitri Kanevsky, Arthur Nadas, et~al.
\newblock Minimum impurity partitions.
\newblock {\em The Annals of Statistics}, 20(3):1637--1646, 1992.

\bibitem{laber2018binary}
Eduardo~S Laber, Marco Molinaro, and Felipe A~Mello Pereira.
\newblock Binary partitions with approximate minimum impurity.
\newblock In {\em International Conference on Machine Learning}, pages
  2860--2868, 2018.

\bibitem{tal2013construct}
Ido Tal and Alexander Vardy.
\newblock How to construct polar codes.
\newblock {\em IEEE Transactions on Information Theory}, 59(10):6562--6582,
  2013.

\bibitem{romero2015decoding}
Francisco Javier~Cuadros Romero and Brian~M Kurkoski.
\newblock Decoding ldpc codes with mutual information-maximizing lookup tables.
\newblock In {\em 2015 IEEE International Symposium on Information Theory
  (ISIT)}, pages 426--430. IEEE, 2015.

\bibitem{kurkoski2014quantization}
Brian~M Kurkoski and Hideki Yagi.
\newblock Quantization of binary-input discrete memoryless channels.
\newblock {\em IEEE Transactions on Information Theory}, 60(8):4544--4552,
  2014.

\bibitem{nguyen2018capacities}
Thuan Nguyen, Yu-Jung Chu, and Thinh Nguyen.
\newblock On the capacities of discrete memoryless thresholding channels.
\newblock In {\em 2018 IEEE 87th Vehicular Technology Conference (VTC Spring)},
  pages 1--5. IEEE, 2018.

\bibitem{strouse2017deterministic}
DJ~Strouse and David~J Schwab.
\newblock The deterministic information bottleneck.
\newblock {\em Neural computation}, 29(6):1611--1630, 2017.

\bibitem{tishby2000information}
Naftali Tishby, Fernando~C Pereira, and William Bialek.
\newblock The information bottleneck method.
\newblock {\em arXiv preprint physics/0004057}, 2000.

\bibitem{verdu1990channel}
Sergio Verdu.
\newblock On channel capacity per unit cost.
\newblock {\em IEEE Transactions on Information Theory}, 36(5):1019--1030,
  1990.

\end{thebibliography}
\end{document}